\documentclass[aps,10pt,pra,superscriptaddress,twocolumn,footinbib,notitlepage,tightenlines]{revtex4-1}

\usepackage{graphicx}
\usepackage[dvipsnames,svgnames]{xcolor}
\usepackage{bm,bbm}%
\usepackage{braket}
\usepackage{amsthm,amsmath,amssymb}
\usepackage{enumerate}
\usepackage{subfig,caption}
\usepackage{booktabs,multirow}
\usepackage{mathtools,bbding}
\usepackage[percent]{overpic}
\usepackage{microtype}
\usepackage{array,adjustbox,afterpage}
 \usepackage[T1]{fontenc}
 \usepackage[utf8]{inputenc}
 \usepackage{tikz}
 \usetikzlibrary{calc}
 
 \usepackage[english]{babel}

\usepackage{newtxtext,newtxmath}

\newtheoremstyle{red}{}{}{\itshape}{}{\color{red!80!black}\bfseries}{.}{ }{}

\definecolor{darkred}{rgb}{0.57,0,0.12}
\usepackage[colorlinks=true, linkcolor=MidnightBlue, urlcolor=darkred!75!black, citecolor=MidnightBlue]{hyperref}

\usepackage{cleveref}
\usepackage[most,breakable]{tcolorbox}

\AtBeginDocument{
\heavyrulewidth=.08em
\lightrulewidth=.05em
\cmidrulewidth=.03em
\belowrulesep=.65ex
\belowbottomsep=0pt
\aboverulesep=.4ex
\abovetopsep=0pt
\cmidrulesep=\doublerulesep
\cmidrulekern=.5em
\defaultaddspace=.5em
}

\let\nc\newcommand

\hypersetup{
    bookmarksnumbered=true, 
    breaklinks=true,
    unicode=false, 
    pdfstartview={FitH}, 
    pdftitle={}, 
    pdfnewwindow=true, 
    colorlinks=true, 
    allcolors=MidnightBlue!70!black!70!TealBlue
}

\DeclareMathOperator{\Tr}{Tr}

\renewcommand{\H}{\mathcal{H}}

\newcommand{\R}{\mathcal{R}}

\newcommand{\T}{\mathcal{T}}
\newcommand{\B}{\mathcal{B}}

\let\LLL\L

\newcommand{\E}{\mathcal{E}}

\newcommand{\F}{\mathcal{F}}

\renewcommand{\L}{\mathcal{L}}

\renewcommand{\P}{\mathcal{P}}
\newcommand{\M}{\mathcal{M}}
\newcommand{\N}{\mathcal{N}}

\newcommand{\Q}{\mathcal{Q}}

\let\ox\otimes

\renewcommand{\*}{\textup{*}}
\newcommand{\<}{\left\langle}
\renewcommand{\>}{\right\rangle}

\renewcommand{\bar}{\;\rule{0pt}{9.5pt}\right|\;}
\let\sbar\bar

\newcommand{\lset}{\left\{\left.}
\newcommand{\rset}{\right\}}

\newcommand{\RR}{\mathbb{R}}

\renewcommand{\SS}{\mathbb{S}}

\newcommand{\OO}{\mathbb{O}}

\newcommand{\FF}{\mathbb{F}}

\newcommand{\mfR}{\mathfrak{R}}
\newcommand{\mfC}{\mathfrak{C}}
\newcommand{\mfM}{\mathfrak{M}}

\newcommand{\id}{\mathbbm{1}}
\newcommand{\idc}{\mathrm{id}}

\newenvironment{boxxed}[1]%
  {\expandafter\ifstrequal\expandafter{#1}{red}{\begin{tcolorbox}[colback=orange!3,colframe=orange!15]}{\begin{tcolorbox}[colback=white,colframe=gray!10,breakable,enhanced]}}%
  {\end{tcolorbox}}

\let\textsc\uppercase

{\begin{boxxed}{orange}\vspace{-\baselineskip}\begin{equation}}%
{\end{equation}\end{boxxed}}

\RenewDocumentEnvironment{boxed}{}%
{\begin{boxxed}{white}}%
{\end{boxxed}}

\newtheorem{theorem}{Theorem}

\newtheorem{corollary}[theorem]{Corollary}

\newtheorem{lemma}[theorem]{Lemma}

\theoremstyle{definition}
\newtheorem*{remark}{Remark}

\theoremstyle{red}

\let\oldproofname\proofname
\renewcommand{\proofname}{\rm\bf{\oldproofname}}

  \nc{\MIO}{{\text{\rm MIO}}}
\nc{\DIO}{{\text{\rm DIO}}}
\nc{\SIO}{{\text{\rm SIO}}}
\nc{\IO}{{\text{\rm IO}}}
\nc{\CRNG}{{\text{\rm CRNG}}}

\nc{\lsetr}{\left\{\,}
\nc{\rsetr}{\right.\right\}}
\nc{\barr}{\,\rule{0pt}{9.5pt}\left|\;}
\nc{\ketbra}[2]{\ket{#1}\!\bra{#2}}

\nc{\TT}{\mathbb{T}}
\nc{\CT}{\Upsilon}

\nc{\wt}{\widetilde}

\nc{\dia}{\!\!\Diamond}
\nc{\bdia}{\!\!\vardiamond}

\DeclareMathOperator{\aff}{aff}

\newcommand{\RFg}{R_{\max,\FF}}
\newcommand{\RFs}{R_{s,\FF}}
\newcommand{\RFmin}{R_{\min,\FF}}
\newcommand{\ROg}{R_{\max,\OO}}
\newcommand{\ROs}{R_{s,\OO}}
\newcommand{\ROmin}{R_{\min,\OO}}

\newcommand{\bal}{\begin{equation}\begin{aligned}}
\newcommand{\eal}{\end{aligned}\end{equation}}
\newcommand{\dm}[1]{\ketbra{#1}{#1}}

\makeatletter
\def\renewtheorem#1{%
  \expandafter\let\csname#1\endcsname\relax
  \expandafter\let\csname c@#1\endcsname\relax
  \gdef\renewtheorem@envname{#1}
  \renewtheorem@secpar
}
\def\renewtheorem@secpar{\@ifnextchar[{\renewtheorem@numberedlike}{\renewtheorem@nonumberedlike}}
\def\renewtheorem@numberedlike[#1]#2{\newtheorem{\renewtheorem@envname}[#1]{#2}}
\def\renewtheorem@nonumberedlike#1{  
\def\renewtheorem@caption{#1}
\edef\renewtheorem@nowithin{\noexpand\newtheorem{\renewtheorem@envname}{\renewtheorem@caption}}
\renewtheorem@thirdpar
}
\def\renewtheorem@thirdpar{\@ifnextchar[{\renewtheorem@within}{\renewtheorem@nowithin}}
\def\renewtheorem@within[#1]{\renewtheorem@nowithin[#1]}
\makeatother

\nc{\phig}{{\phi_{\rm gold}}}

\begin{document}


 \title{One-Shot Manipulation of Dynamical Quantum Resources}

 \author{Bartosz Regula}
\thanks{Both authors contributed equally to this work.}
\affiliation{School of Physical and Mathematical Sciences, Nanyang Technological University, 637371, Singapore}
\email{bartosz.regula@gmail.com}
\author{Ryuji Takagi}
\thanks{Both authors contributed equally to this work.}
\affiliation{School of Physical and Mathematical Sciences, Nanyang Technological University, 637371, Singapore}
\affiliation{Center for Theoretical Physics and Department of Physics, Massachusetts Institute of Technology, Cambridge, Massachusetts 02139, USA}
\email{ryuji.takagi@ntu.edu.sg}

\begin{abstract}
We develop a unified framework to characterize one-shot transformations of dynamical quantum resources in terms of resource quantifiers, establishing universal conditions for exact and approximate transformations in general resource theories. Our framework encompasses all dynamical resources represented as quantum channels, including those with a specific structure --- such as boxes, assemblages, and measurements --- thus immediately applying in a vast range of physical settings. For the particularly important manipulation tasks of distillation and dilution, we show that our conditions become necessary and sufficient for broad classes of important theories, enabling an exact characterization of these tasks and establishing a precise connection between operational problems and resource monotones based on entropic divergences. We exemplify our results by considering explicit applications to: quantum communication, where we obtain exact expressions for one-shot quantum capacity and simulation cost assisted by no-signalling, separability-preserving, and positive partial transpose--preserving codes; as well as to nonlocality, contextuality, and measurement incompatibility, where we present operational applications of a number of relevant resource measures.
\end{abstract}

\maketitle


\textit{\textbf{Introduction.}}
---
Manipulating different resources under physical restrictions underpins many quantum technologies, and the precise understanding of physically realizable transformations is a fundamental question both from theoretical and practical points of view. 
Quantum resource theories~\cite{horodecki_2012,coecke_2016,chitambar_2019} provide a platform through which the quantification and manipulation of resources can be explicitly considered, enabling the study of resource manipulation in a range of settings of interest~\cite{horodecki_2009,streltsov_2017,gour_2008,brandao_2015}.  

Due to the inherent generality of the framework of resource theories, the formalism could be expected to provide unifying results on resource manipulation that hold for diverse classes of resources.
Developing such versatile \textit{general resource theories} allows one to extract common features shared by different physical phenomena and clarify the peculiarities that differentiate one setting from others~\cite{brandao_2015,liu_2017,gour_2017,anshu_2018-1,regula_2018,lami_2018,takagi_2019-2,takagi_2019,uola_2019-1,liu_2019,Takagi2020universal,regula_2020,fang_2020,Zhou2020general,Vijayan2020oneshot,regula_2020-2,fang_2020-2,kuroiwa_2020,Regula2021operational}. 
However, previous approaches to these problems suffer from a number of limitations. 
Many works have focused on the static, rather than dynamical resources, in the sense that only the manipulation of quantum states was considered, and a more general approach that incorporates the ability to manipulate the dynamics of the systems was not established. Recent works have begun to describe channel-based theories~\cite{takagi_2019,liu_2020,liu_2019-1,gour_2019-1,regula_2020-2,fang_2020-2,takagi_2020,takagi2020optimal}, but often focused on the investigation of specific theories such as entanglement, coherence, or quantum memories~\cite{bendana_2017,diaz_2018-2,rosset_2018,bauml_2019,Gour2020dynamical,saxena_2020,pirandola_2017,faist_2018,theurer_2019,Theure2020dynamical,yuan_2020,wang_2019-1,Wang2019distinguish_channels}, or obtained results that only apply in the idealized asymptotic limit~\cite{liu_2020}.
The characterization of dynamical resource theories is significantly more complex than the manipulation of underlying states, and many questions remain unanswered. Notably, the precise characterization of convertibility between two quantum channels in the practical \textit{one-shot} scenario has been an outstanding problem to be addressed. 

An important aspect of dynamical resources is that they can describe a much broader range of settings than the commonly considered manipulation of quantum states with general quantum channels. 
For instance, Bell nonlocality~\cite{bell1964einstein,Brunner2014Bell} as well as quantum contextuality~\cite{KOCHEN1967,Cabello2014graph} have been investigated within formal resource-theoretic settings~\cite{Geller_2014,de_Vicente_2014_nonlocality,Wolfe2020quantifyingbell,Horodecki2015axiomatic,Duarte2018contextuality}, but the specific restrictions on the structure of channels allowed in these settings prevented them from being integrated into most of the previous general quantum resource frameworks.  
The recent works of Refs.~\cite{Schmid2020typeindependent,Rosset2020type} considered an approach to quantum nonlocality that encompasses any type of input and output systems, suggesting that the establishment of a broader framework of channel resource theories might be possible, and opening up the potential to address the manipulation of dynamical resources in a unified manner. 

Here, we achieve such a unified description by providing a universal characterization of one-shot resource transformations with finite error in terms of fundamental resource measures, valid in general resource theories of quantum channels. 
By allowing arbitrary restrictions on the set of channels under consideration, our results apply not only to settings previously studied in resource theories of channels --- e.g., entanglement, coherence, magic, quantum thermodynamics, and quantum communication --- but also to other dynamical resources such as Bell nonlocality, contextuality, steering~\cite{Gallego2015steering,Uola2020steering}, measurement incompatibility~\cite{Heinosaari2016invitation}, and many others~\cite{Schmid2020typeindependent, Rosset2020type}, offering a very general quantitative description of the fundamental task of resource manipulation.
As an important subclass of manipulation tasks, we study resource distillation and dilution, where we present additional bounds and results that provide necessary and sufficient conditions for resource conversion in many relevant cases. This provides an exact characterization of optimal one-shot rates achievable in these tasks, in which case important resource measures --- such as resource robustness and hypothesis testing relative entropy --- are endowed with explicit operational meaning.
We further establish tight benchmarks on the achievable fidelity of distillation.
Finally, we discuss insights provided by our general results into several physical scenarios.
We first apply our results to quantum communication and find exact expressions for quantum capacity and simulation cost for communication assisted by no-signalling codes and codes preserving separability or the positivity of the partial transpose (PPT).
We then discuss the application to nonlocality, contextuality, and measurement incompatibility, where we link resource measures previously introduced in other settings to approximate resource transformations.

We focus on discussing our main results below, and the technical proofs are deferred to the Supplemental Material~\footnote{See the Supplemental Material for detailed proofs of our main results, which includes Refs.~\cite{sion_1958,fuchs_1999,kretschmann_2004,johnston_2018,brandao_2010,brandao_2011}}.

\nocite{sion_1958,fuchs_1999,kretschmann_2004,johnston_2018,brandao_2010,brandao_2011}

\textit{\textbf{Manipulation of dynamical resources.}}
---
Let $\OO_{\rm all}$ be the set of valid channels allowed in the given physical setting; this can be the set of all quantum channels or a subset thereof, allowing us to take into consideration possible restrictions on the types of channels being manipulated.
Each resource theory also designates a subset of channels that are considered to be available for free, and we denote the given free channels as $\OO\subseteq\OO_{\rm all}$.
We impose mild assumptions that the underlying Hilbert spaces are always finite dimensional and, for a fixed dimension, the set of free channels $\OO$ is convex and closed~\cite{brandao_2015}. 

General transformations of quantum channels are described by quantum superchannels~\cite{chiribella_2008}.
Since superchannels need not preserve specific channel structures in general~\cite{Schmid2020typeindependent}, we consider the subset of superchannels that map the set of allowed channels $\OO_{\rm all}$ to allowed output channels $\OO_{\rm all}'$, defined as $\SS_{\rm all}\coloneqq\{\Theta:\OO_{\rm all}\rightarrow\OO_{\rm all}'\}$. 
We can now take a subset of $\OO_{\rm all}'$ and consider it as the free channels in the output space, which we denote by $\OO'$.
A subset of $\SS_{\rm all}$ serves as the set of free transformation that can be used for the manipulation of resources. 
The standard requirement for any free operation is that it should not generate any resource, i.e., it should not create any costly channel out of a free one.
We consider the maximal set satisfying this condition, $\SS\coloneqq\lset\Theta\in\SS_{\rm all}\sbar\Theta(\M)\in\OO'\;\forall \M\in\OO\rset$.  

Our goal is to find the conditions for the transformation from $\E$ to $\N$ using free superchannels in $\SS$, given any two channels $\E\in\OO_{\rm all}$ and $\N\in\OO_{\rm all}'$.
In practice, the transformation can often only be achieved approximately, especially in nonasymptotic resource manipulation.
To evaluate the inevitable error, we consider the worst-case fidelity~\cite{belavkin_2005,gilchrist_2005} defined for two channels $\E_1,\E_2\in\OO_{\rm all}$ as $F(\E_1,\E_2)\coloneqq\min_{\rho}F(\idc\otimes\E_1(\rho),\idc\otimes\E_2(\rho))$ where $F(\rho,\sigma)\coloneqq \|\sqrt{\rho}\sqrt{\vphantom{\rho}\sigma}\|_1^2$ is the fidelity.

Our aim will be to characterize the conditions on resource transformation through the resource contents of the given channels. To this end, we introduce two types of resource measures. They both belong to, or are closely related to, the class of one-shot entropic quantities~\cite{renner_2005} and in particular channel divergences~\cite{yuan_2019,gour_2020}.
The first class is known as the robustness measures~\cite{vidal_1999} defined for any $\E\in\OO_{\rm all}$ as
\bal
R_{\OO;\tilde\OO}(\E)\coloneqq\inf\lset 1+r \sbar \frac{\E+r\M}{1+r}\in\OO,\M\in\tilde\OO\rset,
\label{eq:robustness def}
\eal
where $\tilde\OO\subseteq \OO_{\rm all}$ is some set of channels containing $\OO$. The extreme case $\tilde\OO =\OO_{\rm all}$ is known as \textit{generalized robustness}~\cite{steiner_2003,harrow_2003,takagi_2019} and corresponds to the max-relative entropy~\cite{datta_2009}, hence we denote it by $R_{\max,\OO} \coloneqq R_{\OO;\OO_{\rm all}}$. The other case of interest is the  \textit{standard robustness} $R_{s,\OO} \coloneqq R_{\OO;\OO}$~\cite{yuan_2020,takagi2020optimal}.
We also define the smooth robustness $R_{\OO;\tilde\OO}^\epsilon(\E)\coloneqq \min\{R_{\OO;\tilde\OO}(\E')|F(\E',\E)\geq 1-\epsilon, \E'\in\OO_{\rm all}\}$ for $0\leq \epsilon\leq 1$.
The other type of measure, based on the hypothesis testing relative entropy~\cite{buscemi_2010-1,wang_2012,yuan_2019,gour_2020}, is defined for $\E\in\OO_{\rm all}$ as
\bal
R_{H,\OO}^\epsilon(\E)\coloneqq \min_{\M\in\OO}\max_{\psi}R_H^\epsilon(\idc\otimes\E(\psi)||\idc\otimes\M(\psi))
\eal
where $R_H^\epsilon(\rho||\sigma)\coloneqq  \max_{\substack{0\leq P \leq \id\\\Tr(P\rho)\geq 1-\epsilon}}\Tr(P\sigma)^{-1}$ and the optimization is restricted to pure input states $\psi$ without loss of generality. This entropic quantity characterizes the distinguishability between $\E$ and the channels in the set $\OO$. 
The case of $\epsilon=0$ is known as the min-relative entropy~\cite{datta_2009}, denoted as $R_{\min,\OO}(\E)$.

It is also useful to introduce two classifications for the given theory depending on the properties of $\OO$. 
We say that $\OO$ is \textit{full dimensional} if $\mathrm{span}(\OO)$ contains all channels in $\OO_{\rm all}$, and \textit{reduced dimensional} otherwise~\cite{regula_2020}. 
Intuitively, in full-dimensional theories, the set $\OO$ is of full measure, meaning that $R_{s,\OO}(\E) < \infty$ for all $\E$; examples include the theory of entanglement or local operations and shared randomness. 
On the other hand, reduced-dimensional theories are equipped with a set of free channels $\OO$ of zero measure, and the standard robustness $R_{s,\OO}$ can diverge, examples of which include the theory of coherence, asymmetry, and quantum thermodynamics. 
In order to characterize such resources, we will often need to consider an optimization with respect to $\aff(\OO)$, the affine hull of $\OO$~\cite{regula_2017,regula_2020,note_affine}, and we define in particular $R_{H,\aff(\OO)}^\epsilon$ (respectively, $R_{\min,\aff(\OO)}$) as the hypothesis testing entropy (min-relative entropy) minimized over $\aff(\OO)$ instead of $\OO$.

We can now state our main results which connect the resource monotones with general one-shot resource transformations.

\begin{theorem}
\label{thm:transformation general} 
Let $\E\in\OO_{\rm all}$ and $\N\in\OO_{\rm all}'$.
If there exists a free superchannel $\Theta\in\SS$ such that $F(\Theta(\E),\N)\geq 1-\epsilon$, then for any monotone $\mathfrak{R}_\OO$ it holds that $\mathfrak{R}_{\OO}(\E)\geq \mathfrak{R}_{\OO'}^\epsilon(\N)$, as well as that $\mathfrak{R}_{\OO}^\delta(\E)\geq \mathfrak{R}_{\OO'}^{2(\sqrt{\delta}+\sqrt{\epsilon})}(\N)$ for any $0\leq \delta \leq 1$ where $\mfR_\OO^\epsilon(\E)\coloneqq \min\{\mfR_\OO(\E')\,|\,F(\E',\E)\geq 1-\epsilon,\,\E'\in\OO_{\rm all}\}$.

Conversely, for any choices of $\epsilon, \delta \geq 0$ such that $\epsilon + 2\delta < 1$, there exists a free superchannel $\Theta\in\SS$ such that $F(\Theta(\E),\N)\geq 1 -\epsilon - 2\delta$ if $R_{H,\OO}^\delta(\E)\geq R_{s,\OO'}^\epsilon(\N)$ or if $R_{H,\aff(\OO)}^\delta(\E)\geq R_{\max,\OO'}^\epsilon(\N)$.
\end{theorem}

Here, the parameters $\epsilon, \delta$ can be used to study the trade-offs between the error allowed in the transformation and the values of the smoothed monotones $\mfR_\OO^\epsilon$ of both the input and the output channel.

Notice that we provided two alternative achievability conditions: one using the hypothesis testing measure $R_{H,\OO}$ and the standard robustness $R_{s,\OO'}$, and one using the affine hypothesis testing measure $R_{H,\aff(\OO)}$ and the robustness $R_{\max,\OO'}$. The reason for this is that the former condition will typically trivialize in reduced-dimensional theories, while the latter condition trivializes for full-dimensional theories.

Thm.~\ref{thm:transformation general} establishes the conditions for general resource transformation universally applicable to any resource theory, including ones with specific structures of allowed channels, reflected by choosing appropriate sets $\OO_{\rm all}$ and $\OO_{\rm all}'$.
In the special case of manipulating quantum states (channels with trivial input), we recover the results of Ref.~\cite{Zhou2020general} and extend them to reduced-dimensional theories.
We stress that the monotones are all convex optimization problems and they reduce to computable semidefinite programs when $\OO$ is characterized by semidefinite constraints~\cite{Note1}.

\textit{\textbf{Distillation and dilution}}
---
Two of the most important classes of resource transformation tasks are resource distillation, where general resources are transformed into some reference target resources, and dilution, where the reference target resources are used to synthesize a given channel through free transformations. 
In particular, one is often interested in two quantities: the \textit{distillable resource} $d^\epsilon_{\OO} (\E)$ and the \textit{resource cost} $c^\epsilon_\OO(\E)$; choosing a suitable class of target reference channels $\TT \subseteq \OO'_{\rm all}$, we can define
\bal
  d^\epsilon_{\OO} (\E)\! &\coloneqq \sup \lset \mathfrak{R}_{\OO'}(\T) \!\bar\! F(\Theta(\E), \T) \geq 1\!-\!\epsilon,\; \T\! \in \TT,\; \Theta \in \SS \rset\\
  c^\epsilon_{\OO} (\E)\! &\coloneqq \inf \lset \mathfrak{R}_{\OO'}(\T) \!\bar\! F(\E, \Theta(\T)) \geq 1\!-\!\epsilon,\; \T\! \in \TT,\; \Theta \in \SS \rset
\eal
where $\mathfrak{R}_{\OO'}$ refers to any chosen monotone --- for example, $R_{\min,\OO'}$, $R_{\max,\OO'}$, or $R_{s,\OO'}$. In the discussion below, we will fix $\mathfrak{R}_{\OO'} = R_{\min,\OO'}$ for simplicity. 
The target channels $\TT$ are often chosen as multiple copies of some fixed reference channel, but we allow for broader types of targets.

Notably, under suitable conditions on the reference channels $\TT$, the necessary and sufficient conditions of Thm.~\ref{thm:transformation general} coincide, yielding a precise characterization of the resource cost.
\begin{corollary}\label{cor:dilution}
If the chosen reference set satisfies $R_{\min,\OO'}(\T) = R_{s,\OO'}(\T) \; \forall \T \in \TT$,
then it holds that $ c^\epsilon_{\OO} (\E) = \lceil\ROs^\epsilon(\E)\rceil_{\TT}$.

Similarly, if the chosen reference set obeys $R_{\min,\aff(\OO')}(\T) = R_{\max,\OO'}(\T) \; \forall \T \in \TT$,
then it holds that $ c^\epsilon_{\OO} (\E) = \lceil\ROg^\epsilon(\E)\rceil_{\TT}$.
\end{corollary}
Here, we used the notation $\lceil\cdot\rceil_{\TT}$ to indicate the smallest value greater than or equal to the argument for which there exists a corresponding channel $\T \in \TT$ --- this is required, for instance, when $\TT$ forms a discrete set (see Ref.~\cite{Note1}).

The result establishes an operational meaning of the measures $\ROs^\epsilon$ or $\ROg^\epsilon$ as long as the conditions are satisfied. This raises the question: when does a choice of reference channels $\TT$ satisfying $\ROmin (\T) = \ROs (\T)$ or $R_{\min,\aff(\OO)} (\T) = \ROg (\T)$ exist? It emerges that this is a commonly occurring feature of resource theories. For instance, it is satisfied by relevant choices of target channels in resource theories of quantum memories and communication, immediately providing a characterization of one-shot simulation cost of channels. When discussing the transformations of quantum states, it was shown that in \textit{any} convex resource theory there exists maximal ``golden states'' $\phi_{\rm gold}$ such that $R_{\min,\OO}(\phi_{\rm gold})=R_{\max,\OO}(\phi_{\rm gold})$~\cite{regula_2020}, and indeed such states satisfy all requirements of the Corollary in theories such as quantum entanglement or quantum coherence. In the case of quantum state manipulation, our result recovers the considerations of Ref.~\cite{liu_2019}, where a general framework for quantum state resources was established.

We now turn to the case of distillation. Here, we can improve the bound in Thm.~\ref{thm:transformation general} and obtain an alternative necessary condition. Importantly, distillation is often understood as the \textit{purification} of noisy resources, in which case it is natural to consider pure reference channels $\TT$ --- for instance, if the input space and output spaces coincide, unitary channels can serve as targets, while if the input space is trivial, then pure-state preparation channels can be regarded as targets.
The property needed for $\N$ to serve as the reference resource is that the output states for pure input states remain pure.
In such cases, we obtain the following general conditions.

\begin{theorem}
\label{thm:general distillation}
Let $\E\in\OO_{\rm all}$, $\N\in\OO_{\rm all}'$ and suppose $\idc\otimes\N(\psi)$ is pure for any pure state $\psi$.
If there exists a free superchannel $\Theta\in\SS$ such that $F(\Theta(\E),\N)\geq 1-\epsilon$, then it holds that $R_{H,\OO}^\epsilon(\E)\geq R_{\min,\OO'}(\N)$ and $R_{H,\aff(\OO)}^\epsilon(\E)\geq R_{\min,\aff(\OO')}(\N)$.

Conversely, there exists a free superchannel $\Theta\in\SS$ such that $F(\Theta(\E),\N)\geq 1 - \epsilon -\delta$ if $R_{H,\OO}^\epsilon(\E)\geq R_{s,\OO'}^\delta(\N)$ or if $R_{H,\aff(\OO)}^\epsilon(\E)\geq R_{\max,\OO'}^\delta(\N)$.
\end{theorem}
Thm.~\ref{thm:general distillation} adds a useful alternative characterization for distillation to the general condition provided by Thm.~\ref{thm:transformation general}.
In particular, similarly to the case of dilution, we can obtain the following.
\begin{corollary}\label{cor:distillation}
Consider any reference set $\TT$ such that $\idc \otimes \T(\psi)$ is pure for any pure $\psi$ and any $\T \in \TT$.

If $\TT$ also satisfies $R_{\min,\OO}(\T) = R_{s,\OO}(\T) \; \forall \T \in \TT$,
then it holds that $d^\epsilon_{\OO} (\E) = \lfloor R_{H,\OO}^\epsilon(\E)\rfloor_{\TT}$.

Similarly, if the chosen reference set obeys $R_{\min,\aff(\OO)}(\T) = R_{\max,\OO}(\T) \; \forall \T \in \TT$,
then it holds that $d^\epsilon_{\OO} (\E) = \lfloor R_{H,\aff(\OO)}^\epsilon(\E)\rfloor_{\TT}$.
\end{corollary}
This establishes a precise characterization of distillable resource in any resource theory for suitable target channels $\TT$, and furthermore gives an exact operational meaning to the resource measures $R_{H,\OO}^\epsilon$ and $R_{H,\aff(\OO)}^\epsilon$ in the task of distillation. When the manipulated objects are quantum states, we recover the results of Refs.~\cite{liu_2019,regula_2020}.

In some cases, distillation with the desired precision might not be possible. It is then of interest to instead ask how close one can approximate the chosen target channel, that is, characterize the maximal achievable \textit{fidelity of distillation}. We can adapt our methods to obtain close upper and lower bounds for this quantity. Importantly, the bounds become tight for relevant reference states obeying conditions as in Cor.~\ref{cor:distillation}, allowing us to provide an exact expression for the fidelity of distillation. We discuss the full details in Ref.~\cite{Note1}.

Another bound on distillation fidelity was recently presented with respect to the so-called resource weight~\cite{regula_2020-2,fang_2020-2}. 
Our approach allows one to extend the insight from these works to theories with arbitrary channel structures $\OO_{\rm all}$, enabling an operational application of the corresponding weight measures, some of which were previously introduced in other contexts~\cite{lewenstein_1998,de_Vicente_2014_nonlocality,skrzypczyk_2014,grudka_2014,Abramsky2017fraction,Pusey2015incompatibility}.

One can obtain additional results in the characterization of distillation and dilution in special cases of resource theories, for example, when the given theory is concerned with an underlying state-based resource. We discuss such cases in Ref.~\cite{Note1}.


\textit{\textbf{Quantum communication}}
---
A central problem in quantum communication is to manipulate a given channel to enhance its communication capabilities using resources available to both parties.
Quantum capacity~\cite{bennett_1996,lloyd_1997,bennett_2002} and simulation cost~\cite{bennett_2002,Bennett2014reverse,Berta2011reverse} are important figures of merit to evaluate the operational capability of quantum channels, and their one-shot characterization received considerable attention recently~\cite{Leung2015NS,tomamichel_2016-1,wang_2019-3,Fang2020max,takagi_2020,Hsieh2020communication}.
These tasks are precisely channel distillation and dilution where the reference resource is the identity channel, and the sets of free channels and free superchannels specify the accessible resources for the sender and receiver. 
Formally, we define the one-shot quantum capacity and simulation cost with the set of free superchannels $\SS$ as
\bal
Q^\epsilon_\SS(\E)&\coloneqq\max\lset\log d\sbar \exists \Theta\in\SS, F(\Theta(\E),\idc_d)\geq 1-\epsilon\rset\\
C^\epsilon_\SS(\E)&\coloneqq\min\lset\log d\sbar \exists \Theta\in\SS, F(\Theta(\idc_d),\E)\geq 1-\epsilon\rset.
\eal

We can then use our results to immediately obtain an exact characterization of these quantities in relevant settings. For example, when $\OO$ is the set of separable channels $\OO_{\rm SEP}$ whose Choi states are separable, this setting corresponds to communication with separability-preserving codes $\SS = \SS_{\rm SEP}$.
This is a full-dimensional theory and it holds that $R_{\min,\OO_{\rm SEP}}(\idc_d)=R_{s,\OO_{\rm SEP}}(\idc_d)=d$~\cite{yuan_2020}.
Then, Cors.~\ref{cor:dilution} and~\ref{cor:distillation} provide a complete characterization of one-shot quantum capacity and simulation cost as $Q_{\SS_{\rm SEP}}^\epsilon(\E)=\log \lfloor R_{H,\OO_{\rm SEP}}^\epsilon(\E) \rfloor$ and $C_{\SS_{\rm SEP}}^\epsilon(\E)=\log \lceil R_{s,\OO_{\rm SEP}}^\epsilon(\E) \rceil$, the latter of which recovers a result of Ref.~\cite{yuan_2020}. Similar results apply to the setting of communication assisted by codes preserving the positivity of the partial transpose (PPT), where $\OO_{\rm PPT}$ is the set of PPT channels~\cite{rains_2001,Leung2015NS}. We analogously obtain $Q_{\SS_{\rm PPT}}^\epsilon(\E)=\log\lfloor R_{H,\OO_{\rm PPT}}^\epsilon(\E) \rfloor$ and $C_{\SS_{\rm PPT}}^\epsilon(\E)=\log \lceil R_{s,\OO_{\rm PPT}}^\epsilon(\E) \rceil$. Interestingly, $R_{H,\OO_{\rm PPT}}^\epsilon$ appeared as a bound in Ref.~\cite{tomamichel_2016-1}.

Another example is the case when $\OO$ is the set of replacement channels $\OO_R\coloneqq\{\R_\sigma|\R_\sigma(\cdot)=\Tr(\cdot)\sigma\}$, where $\SS$ becomes the set of channel transformations assisted by no-signalling (NS) correlations, $\SS_{\rm NS}$~\cite{Duan2016nosignalling,takagi_2020}.
Since $\OO_R$ is closed under linear combinations, it is reduced dimensional.   
We also have $R_{\min,\aff({\OO_{R}})}(\idc_d)= R_{\max,\OO_{R}}(\idc_d)=d^2$, and our results
immediately give $Q_{\SS_{\rm NS}}^\epsilon(\E)=\frac{1}{2}\log \lfloor R_{H,\aff(\OO_R)}^\epsilon(\E)\rfloor$ and $C_{\SS_{\rm NS}}^\epsilon(\E)=\frac{1}{2}\log \lceil R_{\max,\OO_{R}}^\epsilon(\E) \rceil$.
The one-shot NS-assisted quantum capacity was obtained in Ref.~\cite{wang_2019-3} in the form of a semidefinite program; our result identifies it with the affine hypothesis testing relative entropy, providing an operational meaning to this resource measure. 
The one-shot NS simulation cost was obtained in Refs.~\cite{takagi_2020,Fang2020max}, which is recovered by our general approach as a special case. Furthermore, we can quantify exactly the fidelity of NS-assisted coding~\cite{Note1}, recovering a result of Ref.~\cite{Leung2015NS}.

Our methods apply also to the study of the entanglement of bipartite channels~\cite{gour_2019,bauml_2019}, where the target resources are maximally entangled states in the underlying state-based resource theory. We then establish exact one-shot rates of channel manipulation under PPT- and separability-preserving superchannels (see Ref.~\cite{Note1}).



\textit{\textbf{Nonlocality and contextuality}}
---
Quantum nonlocality has been a major subject of study not only as a key feature of quantum theory, but also as a useful resource in a number of operational tasks~\cite{Buhrman2010nonlocality,Cubitt2010zeroerror,Barrett2005key,Acin2006from,pironio_random_2010,Vazirani2014fully}.
The latter view motivates a precise understanding of the manipulation of nonlocal resources~\cite{de_Vicente_2014_nonlocality,Forster2009nonlocality,Nery2020steering,Gallego2015steering}, but the characterization of such one-shot transformations has remained elusive. 
Our framework encompasses this scenario by choosing $\OO_{\rm all}$ to be the set of no-signalling channels, where classical input or output systems are represented through dephasing in a given basis, and taking $\OO$ to be the channels that can be constructed by local operations and shared randomness. This includes not only the standard setting of Bell nonlocality (where such channels are the classical ``boxes''), but also more general resources such as steering (where channels represent ``assemblages'')~\cite{Schmid2020typeindependent,Rosset2020type}.
Our results then provide an operational application of nonlocality measures individually introduced in different settings of nonlocality~\cite{de_Vicente_2014_nonlocality,Piani2015steering} and unified in Refs.~\cite{Schmid2020typeindependent,Rosset2020type}, which we relate to one-shot resource transformations (see also Ref.~\cite{Wolfe2020quantifyingbell}). We also introduce new monotones to this setting, which can add further insights into feasible manipulation of nonlocal resources. 
For instance: since a noisy box of the form $\B = (1-\epsilon)\B_{\rm PR} + \epsilon \B_{\rm free}$, where $\B_{\rm PR}$ is the Popescu-Rohrlich (PR) box~\cite{Popescu1994PR} and $\B_{\rm free}\in\OO$ is some local box, has $R_{\min,\OO}(\B^{\otimes n}) = 1 \; \forall n$, we recover the fact that it cannot be distilled to other fundamental resources such as the PR box even when multiple copies of the box $\B$ are available~\cite{dukaric2008limit,Forster2009nonlocality}.
In Ref.~\cite{Note1}, we numerically evaluate the resource measures for a special class known as isotropic boxes~\cite{Masanes2006isotropic} where we explicitly observe this property.

Notably, our framework can also be applied to another related phenomenon known as quantum contextuality, which also serves as an operational resource~\cite{Hameedi2017communication,Raussendorf2013contextuality,Howard204contextuality,Bermejo-Vega2017contextuality,Kleinmann2011memory,Cabello2018memory,Schmid2018contextual}.
Namely, we consider the set of all classical-classical channels for consistent boxes~\cite{grudka_2014} as $\OO_{\rm all}$ and the set of channels corresponding to noncontextual boxes as $\OO$.
Our results characterize the exact and approximate box transformations with operations that do not create contextuality, offering a new perspective to the recent resource-theoretic framework~\cite{Horodecki2015axiomatic, Duarte2018contextuality,Amaral2018wiring}, as well as providing operational application of the robustness of contextuality~\cite{Meng2016robustness,Meng2016generalized,Li2020contextual}, contextual fraction~\cite{Abramsky2011sheaf,grudka_2014,Abramsky2017fraction}, and hypothesis testing measures introduced in this work. 


\textit{\textbf{Measurement incompatibility}}
---
Measurement incompatibility refers to the impossibility of simultaneous measurement and is closely related to the aforementioned phenomena such as Bell nonlocality and steering~\cite{Quintino2014joint,Uola2014joint,Pusey2015incompatibility}.
The set of POVMs $\{M_{a|x}\}$, where $M_{a|x}$ is the POVM element with outcome $a$ for the measurement labeled by setting $x$, is called compatible (or jointly measurable) if there exists a parent measurement $\{P_i\}$ and a conditional probability distribution $\{q(a|x,i)\}$ such that $M_{a|x}=\sum_i q(a|x,i)P_i$.

Our formalism can handle scenarios where resources take the form of ensembles by incorporating the classical labels of the ensembles into the description of $\OO_{\rm all}$.
In the case of measurement incompatibility, $\OO_{\rm all}$ represents the set of channels corresponding to POVMs, while $\OO\subseteq\OO_{\rm all}$ represents the set of compatible POVMs
(see Ref.~\cite{Note1}).
This form allows one to apply our results to this setting and characterize the approximate one-shot transformation of incompatible sets of measurements, complementing the previous works which focused on exact transformation with smaller sets of free operations in different approaches~\cite{Skrzypczyk2019all,Wenbin2020incompatibility} and providing operational applications of the related measures~\cite{Haapasalo2015robustness,Uola2015robustness,Pusey2015incompatibility,designolle_2019}.
Although the robustness and weight measures are usually defined at the level of POVMs, we show in Ref.~\cite{Note1} that they coincide with the channel-based measures defined in our framework, allowing one to carry over the previous analyses to characterize resource transformations.
We also note that the discussion here can be straightforwardly extended to channel incompatibility~\cite{Heinosaari2016invitation,Heinosaari2017channel}, which includes measurement incompatibility as a special case.  

\textit{\textbf{Conclusions}}
---
We established fundamental bounds on the transformations between general dynamical quantum resources in the one-shot regime. We tightly characterized the ability to manipulate resources by providing conditions for convertibility in terms of the robustness and hypothesis testing measures. In particular, under suitable assumptions, we established an exact quantification of the one-shot distillable resource and one-shot resource cost of general channels, giving a precise operational interpretation to the considered monotones in these important tasks. This not only extends and unifies previous specialized results~\cite{Leung2015NS,wang_2019-3,yuan_2020,gour_2019,bauml_2019,saxena_2020}, but also sheds light on the general structure of dynamical resource theories by providing a common description of their operational aspects. 
Besides contributing to the theory of quantum resources, our methods find direct practical use, as we exemplified with several explicit applications.



\begin{acknowledgments}

\textit{Note --- } During the completion of this manuscript, we became aware of two related works, Ref.~\cite{kim_2020} by Kim et al.\ and Ref.~\cite{yuan_2020-1} by Yuan et al., where the authors independently obtained results overlapping with some of our findings. The former work considers one-shot distillation and dilution of quantum channel entanglement, which coincides with our characterization of these tasks in Corollaries~\ref{cor:dilution} and \ref{cor:distillation} (see Ref.~\cite{Note1}), while the latter work introduces a general framework for one-shot channel distillation and dilution which again corresponds to our Cors.~\ref{cor:dilution} and \ref{cor:distillation} for the cases of quantum channel resources.

We thank David Schmid for sharing useful references with us, which helped us expand the scope of our results. We are also grateful to the authors of Refs.~\cite{kim_2020,yuan_2020-1} for coordinating the submission of our works.
B.R.\ was supported by the Presidential Postdoctoral Fellowship from Nanyang Technological University, Singapore.
R.T.\ acknowledges the support of NSF, ARO, IARPA, AFOSR, the Takenaka Scholarship Foundation, the National Research Foundation (NRF) Singapore, under its NRFF Fellow programme (Award No. NRF-NRFF2016-02), and the Singapore Ministry of Education Tier 1 Grant 2019-T1-002-015. Any opinions, findings and conclusions or recommendations expressed in this material are those of the author(s) and do not reflect the views of National Research Foundation, Singapore.

\end{acknowledgments}

\bibliographystyle{apsrev4-1a}
\bibliography{bib_bartosz,bib_ryuji}

\begin{thebibliography}{130}
\bibitem [{not()}]{note_affine}%
  \BibitemOpen
  \href@noop {} {}\bibinfo {note} {The affine hull $\aff(\OO)$ is the smallest
  affine subspace which contains $\OO$.}\BibitemShut {Stop}%
\end{thebibliography}

\appendix
\widetext

\setcounter{theorem}{4}

\section{Properties of the resource measures}\label{app:properties_measures}

We will use the notation $\<\cdot, \cdot\>$ for the Hilbert-Schmidt inner product $\<A,B\> = \Tr(A^\dagger B)$.

Recall that
\bal
R_{H,\OO}^\epsilon(\E) = \min_{\M\in\OO}\max_{\psi}R_H^\epsilon(\idc\otimes\E(\psi)||\idc\otimes\M(\psi))
\eal
where $R_H^\epsilon(\rho||\sigma)\coloneqq  \max_{\substack{0\leq P \leq \id\\\<P,\rho\>\geq 1-\epsilon}}\<P,\sigma\>^{-1}$. Here, for a channel which takes operators acting on an input Hilbert space $\H_{\rm in}$ to operators acting on  $\H_{\rm out}$, the optimization $\max_{\psi}$ means that we are optimizing over all pure states acting on the Hilbert space $\H_{\rm in'} \ox \H_{\rm in}$ with $\H_{\rm in'} \cong \H_{\rm in}$.

While the quantities $R_{H,\OO}^\epsilon$, $G_{\OO}$, and their affine equivalents might seem like daunting optimization tasks, they actually admit a number of useful properties and can be expressed as convex optimization problems. To see this, we first establish a useful minimax result.
\begin{lemma}\label{lem:minimax}
 For any closed and convex set $\OO$, it holds that
\begin{equation}\begin{aligned}
	R_{H,\OO}^\epsilon(\E)^{-1} &= \min_{\psi} \lset \lambda \bar \< P, \idc \otimes \M (\psi)\> \leq \lambda \; \forall \M\in \OO,\; 0\leq P \leq \id,\; \< P, \idc \otimes \E (\psi)\> \geq 1-\epsilon \rset\\
	R_{H,\aff(\OO)}^\epsilon(\E)^{-1} &= \min_{\psi} \lset \lambda \bar \< P, \idc \otimes \M (\psi)\> = \lambda \; \forall \M\in \OO,\; 0\leq P \leq \id,\; \< P, \idc \otimes \E (\psi)\> \geq 1-\epsilon \rset.
\end{aligned}\end{equation}
\end{lemma}
\begin{proof}
By definition, we have that
\begin{align}
  R_{H,\OO}^\epsilon(\E)^{-1} &=  \max_{\M\in \OO} \min_{\psi}\min_{\substack{0\leq P \leq \id\\ \< P, \idc \otimes \E (\psi)\> \geq 1-\epsilon}} \< P, \idc \otimes \M (\psi)\>\\
  &=  \min_{\psi}\min_{\substack{0\leq P \leq \id\\ \< P, \idc \otimes \E (\psi)\> \geq 1-\epsilon}} \max_{\M\in \OO} \< P, \idc \otimes \M (\psi)\>\\
  &= \min_{\psi} \lset \lambda \bar \< P, \idc \otimes \M (\psi)\> \leq \lambda \; \forall \M\in \OO,\; 0\leq P \leq \id,\; \< P, \idc \otimes \E (\psi)\> \geq 1-\epsilon \rset
\end{align}
where in the second line we used the linearity of $\< P, \idc \otimes \M (\psi)\>$ to apply Sion's minimax theorem~\cite{sion_1958}. In the case of $R_{H,\aff(\OO)}^\epsilon$, the theorem applies analogously ($\aff(\OO)$ is no longer compact, in general, but it still is a convex set) and so we get
\begin{align}
  R_{H,\aff(\OO)}^\epsilon(\E)^{-1} &=  \min_{\psi}\min_{\substack{0\leq P \leq \id\\ \< P, \idc \otimes \E (\psi)\> \geq 1-\epsilon}} \sup_{\M\in \aff(\OO)} \< P, \idc \otimes \M (\psi)\>\\
  &= \min_{\psi} \lset \lambda \bar \< P, \idc \otimes \M (\psi)\> = \lambda \; \forall \M\in \OO,\; 0\leq P \leq \id,\; \< P, \idc \otimes \E (\psi)\> \geq 1-\epsilon \rset
\end{align}
The last line follows from a characterization of affine hulls found e.g.\ in~\cite{regula_2020}, and can be seen explicitly as follows. If there exist $\M_1,\M_2\in\OO$ such that $\<P,\idc\otimes\M_1(\psi)\>\neq \<P,\idc\otimes\M_2(\psi)\>$, then the supremum over $\aff(\OO)$ becomes unbounded. 
On the other hand, the whole quantity is upper bounded by 1 as can be checked by taking $P=\id$. 
Thus, optimal $P$ and $\psi$ necessarily satisfy $\<P,\idc \otimes \M(\psi)\>=\lambda,\forall \M\in\OO$ for some constant $\lambda$, also implying $\sup_{\M\in\aff(\OO)}\<P,\idc \otimes \M(\psi)\>=\lambda$.
\end{proof}
Using the above, we can express the hypothesis testing measures as
\begin{equation}\begin{aligned}
    R_{H,\OO}^\epsilon(\E)^{-1} &= \min_{\rho} \lset \lambda \bar \< Q, J_\M \> \leq \lambda \; \forall \M \in \OO, \; 0 \leq Q \leq \rho \otimes \id,\; \< Q, J_\E \> \geq 1-\epsilon\rset,\\
    R_{H,\aff(\OO)}^\epsilon(\E)^{-1} &= \min_{\rho} \lset \lambda \bar \< Q, J_\M \> =\lambda \; \forall \M \in \OO, \; 0 \leq Q \leq \rho \otimes \id,\; \< Q, J_\E \> \geq 1-\epsilon \rset,
\end{aligned}\end{equation}
where the optimization is over all density matrices $\rho$ acting on the input space $\H_{\rm in}$ of dimension $d_{\rm in}$, and $J_\E = d_{\rm in} (\idc \otimes \E)(\psi^+)$, with $\psi^+$ being the maximally entangled state, is the Choi operator of $\E$. 
Here, we used the fact that any input $\psi$ acting on $\H_{\rm in'} \otimes \H_{\rm in}$ can be written as $\psi = (\rho^{1/2} \otimes \id)  \psi^+ (\rho^{1/2} \otimes \id)$ for some density operator $\rho$ ($\psi$ being the canonical purification of $\rho$). 

In addition to the hypothesis testing measures, we will see later that the following measure characterizes the fidelity of distillation:  
\bal
 G_\OO(\E; m) &= \max_{\psi} \lsetr \< \idc \otimes \E (\psi) ,\, W \> \barr 0 \leq W \leq \id, \< \idc \otimes \M (\psi), W \>  \leq \frac{1}{m} \; \forall \M \in \OO \rsetr,
 \label{eq:G measure def}
\eal
which is parameterized by $m \in \RR$, and we also define $G_{\aff(\OO)}(\E; m)$ analogously. 
In a similar manner, we can also express the measures $G_{\OO}$ and $G_{\aff(\OO)}$ as
\begin{equation}\begin{aligned}
	    G_\OO(\E; m) &= \max_{\rho} \lsetr \< W, J_\E \> \barr 0 \leq W \leq \rho \otimes \id, \;\< W, J_\M \> \leq \frac{1}{m} \; \forall \M \in \OO \rsetr,\\
	    G_{\aff{\OO}}(\E; m) &= \sup_{\rho} \lsetr \< W, J_\E \> \barr 0 \leq W \leq \rho \otimes \id, \;\< W, J_\M \> = \frac{1}{m} \; \forall \M \in \OO \rsetr.
\end{aligned}\end{equation}
To more explicitly see the computability of these problems, let us define the dual cones
\begin{equation}\begin{aligned}
  \OO\* &\coloneqq \lset X \bar \< X, J_\M \> \geq 0 \; \forall \M \in \OO \rset\\
  \aff(\OO)\* &\coloneqq \lset X \bar \< X, J_\M \> = 0 \; \forall \M \in \OO \rset,\\
\end{aligned}\end{equation}
using which we can write
\begin{equation}\begin{aligned}
    R_{H,\OO}^\epsilon(\E)^{-1} &= \min_{\rho} \lset \lambda \bar \lambda \frac{\id}{d_A} - Q \in \OO\*,\; 0 \leq Q \leq \rho \otimes \id,\; \< Q, J_\E \> \geq 1-\epsilon \rset,\\
    G_\OO(\E; m) &= \max_{\rho} \lsetr \< W, J_\E \> \barr 0 \leq W \leq \rho \otimes \id, \; \frac{\id}{d_A} - m W \in \OO\* \rsetr,\\
\end{aligned}\end{equation}
and analogously for $\aff(\OO)$. Whenever the dual cone of operators $\OO\*$ (respectively, $\aff(\OO)\*$) is characterized by semidefinite constraints, the programs $R_{H,\OO}^\epsilon$ and $G_{\OO}$ (respectively, $R_{H,\aff(\OO)}^\epsilon$ and $G_{\aff(\OO)}$) are SDPs, allowing for an efficient evaluation in relevant cases. A similar characterization can be extended to related settings, such as when instead of optimizing over a set of channels $\OO$ one optimizes over a larger set, as is the case when dealing with PPT maps~\cite{Leung2015NS,wang_2019-3}.

We note also that, for $\epsilon = 0$, the hypothesis testing relative entropy $R_{H,\OO}^\epsilon$ reduces to an optimization based on the min-relative entropy~\cite{datta_2009}, given by
\begin{equation}\begin{aligned}
    R_{\min,\OO}(\E)^{-1} &= \min_{\psi} \max_{\M\in \OO} \< \Pi_{\idc \otimes \E(\psi)}, \idc \otimes \M (\psi)\>\\
\end{aligned}\end{equation}
where $\Pi_{\rho}$ denotes the projection onto the support of $\rho$. Of interest to us will be the case when $\idc \otimes \N(\psi)$ is a pure state for any pure input state, in which case
\begin{equation}\begin{aligned}
	R_{\min,\OO}(\N)^{-1} &= \min_{\psi} \lset \lambda \bar \< \idc \otimes \N(\psi), \idc \otimes \M (\psi)\> \leq \lambda \; \forall \M\in \OO \rset
\end{aligned}\end{equation}
We analogously define the min-relative entropy with respect to $\aff(\OO)$, given by
\begin{equation}\begin{aligned}
	R_{\min,\aff(\OO)}(\N)^{-1} &= \inf_{\psi} \lset \lambda \bar \< \idc \otimes \N(\psi), \idc \otimes \M (\psi)\> = \lambda \; \forall \M\in \OO \rset
\end{aligned}\end{equation}
when $\idc \otimes \N(\psi)$ is pure, where we take $\inf \emptyset = \infty$.


\section{Proofs of the main results}\label{app:proofs}

{
\renewcommand\thetheorem{1}
\begin{theorem}
Let $\E\in\OO_{\rm all}$ and $\N\in\OO_{\rm all}'$.
If there exists a free superchannel $\Theta\in\SS$ such that $F(\Theta(\E),\N)\geq 1-\epsilon$, then for any monotone $\mathfrak{R}_\OO$ it holds that $\mathfrak{R}_{\OO}(\E)\geq \mathfrak{R}_{\OO'}^\epsilon(\N)$, as well as that $\mathfrak{R}_{\OO}^\delta(\E)\geq \mathfrak{R}_{\OO'}^{2(\sqrt{\delta}+\sqrt{\epsilon})}(\N)$ for any $0\leq \delta \leq 1$ where $\mfR_\OO^\epsilon(\E)\coloneqq \min\{\mfR_\OO(\E')\,|\,F(\E',\E)\geq 1-\epsilon,\,\E'\in\OO_{\rm all}\}$.

Conversely, for any choices of $\epsilon, \delta \geq 0$ such that $\epsilon + 2\delta < 1$, there exists a free superchannel $\Theta\in\SS$ such that $F(\Theta(\E),\N)\geq 1 -\epsilon - 2\delta$ if $R_{H,\OO}^\delta(\E)\geq R_{s,\OO'}^\epsilon(\N)$ or if $R_{H,\aff(\OO)}^\delta(\E)\geq R_{\max,\OO'}^\epsilon(\N)$.
\end{theorem}
\addtocounter{theorem}{-1}
}
\begin{proof}

To show the first part, we note that the smoothed version of any monotone $\mfR_\OO$ also serves as a  monotone because for any $\Theta\in\SS$, it holds that
\bal
 \mfR_\OO^\epsilon(\E) = \mfR_\OO(\E')\geq \mfR_{\OO'}(\Theta(\E'))\geq \mfR_{\OO'}^\epsilon(\Theta(\E)) 
 \label{eq:monotonicity smooth measure}
\eal
where in the equality we defined $\E'\in\OO_{\rm all}$ as a channel such that $F(\E',\E)\geq 1-\epsilon$ and $\mfR_\OO^\epsilon(\E) = \mfR_\OO(\E')$, and in the first inequality we used the monotonicity of $\mfR_\OO$.
The last inequality is obtained by noting $F(\Theta(\E'),\Theta(\E))\geq F(\E',\E)\geq 1-\epsilon$ due to the data processing inequality of the fidelity under any superchannel~\cite{belavkin_2005,gilchrist_2005}, and by using the definition of the smooth measure. 

Now, suppose $\N'$ is a channel such that $F(\N',\Theta(\E))\geq 1-\delta$ and $\mathfrak{R}_{\OO'}^\delta(\Theta(\E))=\mathfrak{R}_{\OO'}(\N')$. 
Then, we use the relation between the fidelity and diamond norm distance~\cite{fuchs_1999} to get 
\bal
1-F(\N,\N')\leq 1-\left(1-\frac{1}{2}\|\N-\N'\|_\diamond\right)^2\leq \|\N-\N'\|_\diamond &\leq \|\N-\Theta(\E)\|_\diamond+\|\Theta(\E)-\N'\|_\diamond \\
&\leq 2\left(\sqrt{1-F(\N,\Theta(\E))}+\sqrt{1-F(\Theta(\E),\N')}\right) \\
& \leq 2(\sqrt{\epsilon}+\sqrt{\delta}),
\eal
implying that $\mathfrak{R}_{\OO'}(\N')\geq \mathfrak{R}_{\OO'}^{2(\sqrt{\delta}+\sqrt{\epsilon})}(\N)$.
The statement is shown by combining this with \eqref{eq:monotonicity smooth measure}.
When $\delta=0$, $\N'$ coincides with $\Theta(\E)$. Using the assumption that $1-F(\Theta(\E),\N)\leq \epsilon$, we get $\mathfrak{R}_{\OO'}(\N')\geq \mathfrak{R}_{\OO'}^\epsilon(\N)$, resulting in the improved bound.

Let us now show the second part. We start with the condition using $R_{H,\OO}$. 
Let $\tilde\N\in\OO_{\rm all}'$ be a channel such that $F(\tilde\N,\N)\geq 1-\epsilon$ and $R_{s,\OO'}^\epsilon(\N)=R_{s,\OO'}(\tilde\N)$. 
Then, there exists channel $\Q\in\OO'$ such that $\tilde \N + (R_{s,\OO'}^\epsilon(\N)-1)\Q \in R_{s,\OO'}^\epsilon(\N)\OO'$.
Let us also take $P^\star$ and $\psi^\star$ as the optimal operators achieving the optimization in the hypothesis testing entropy measure for $\E$ as $R_{H,\OO}^\delta(\E)=\min_{\M\in\OO}\<P^\star,\idc\otimes\M(\psi^\star)\>^{-1}$ with $0\leq P^\star \leq \id$ and $\<P^\star,\idc\otimes\E(\psi^\star)\>\geq 1-\delta$ (recall Lemma~\ref{lem:minimax}). 
Then, we define the following map:
 \bal
  \Theta(\L)\coloneqq \<P^\star,\idc\otimes\L(\psi^\star)\>\tilde\N +\<\id-P^\star,\idc\otimes\L(\psi^\star)\>\Q. 
  \label{eq:full dimensional free superchannel}
 \eal
The fact that this is a valid superchannel follows since $\Theta$ consists of the preprocessing channel $\omega \mapsto \omega \otimes \psi^\star$,
followed by an application of $\idc \otimes \L$, and the postprocessing channel
\begin{equation}\begin{aligned}
  \omega' \mapsto &\left[ \tilde\N \otimes \< P^\star , \cdot \> \right] (\omega') \\
  + & \left[ \Q \otimes \< \id - P^\star , \cdot \> \right] (\omega'),
\end{aligned}\end{equation}
where the latter can be seen to be a valid CPTP map since $\{P^\star,\id-P^\star\}$ defines a valid POVM. Now, since $\tilde\N,\Q\in\OO_{\rm all}$, it further holds that $\Theta\in\SS_{\rm all}$. 
One can also check that $\Theta\in\SS$ by noting that for any $\M\in\OO$, it holds that $\<P^\star,\idc\otimes\M(\psi^\star)\>\leq R_{H,\OO}^\delta(\E)^{-1}\leq R_{s,\OO'}^\epsilon(\N)^{-1}= R_{s,\OO'}(\tilde\N)^{-1}$, ensuring $\Theta(\M)\in\OO'$.
Using the concavity of the (square root of) fidelity, we obtain 
\bal
 \sqrt{F(\Theta(\E),\N)} &\geq \<P^\star,\idc\otimes\E(\psi^\star)\>\sqrt{F(\tilde\N,\N)} +\<\id-P^\star,\idc\otimes\E(\psi^\star)\>\sqrt{F(\Q,\N)}\\
 &\geq (1-\delta)\sqrt{1-\epsilon},
\eal
which leads to $F(\Theta(\E),\N)\geq(1-\delta)^2(1-\epsilon)\geq 1-\epsilon-2\delta$.

For the condition using $R_{H,\aff(\OO)}$, recall from Lemma~\ref{lem:minimax} that
\begin{align}
  R_{H,\aff(\OO)}^\delta(\E)^{-1} &= \min \Big\{ \lambda \;\Big|\; \< P, \idc \otimes \M (\psi)\> = \lambda \; \forall \M\in \OO, \\
  &\hphantom{= \min \Big\{ \lambda \;\Big|\;} 0\leq P \leq \id,\ \< P, \idc \otimes \E (\psi)\> \geq 1-\delta \Big\}.
  \label{eq:hypothesis_ct_general_input}
\end{align}
Then, let $P^\star$ and $\psi^\star$ be the optimal $P$ and $\psi$ in the above, and $\tilde\N\in\OO_{\rm all}'$ be a channel such that $F(\tilde\N,\N)\geq 1-\epsilon$ and $R_{\max,\OO'}^\epsilon(\N)=R_{\max,\OO'}(\tilde\N)$.
Since $R_{\max,\OO}^\epsilon (\N) \leq R_{H,\aff(\OO)}^\delta (\E)$, there exists a channel $\Q\in\OO_{\rm all}'$ such that $\tilde\N  + (R_{H,\aff(\OO)}^\delta(\E)-1) \, \Q \in R_{H,\aff(\OO)}^\delta(\E) \tilde\OO'$, which allows us to construct the superchannel
\begin{equation}\begin{aligned}
  \Theta(\L) \coloneqq \< P^\star, \idc \otimes \L(\psi^\star) \> \tilde\N + \< \id - P^\star, \idc \otimes \L(\psi^\star)\> \Q.
\end{aligned}\end{equation}
Now, since $\< P^\star, \idc \otimes \M(\psi^\star)\> = R_{H,\aff(\OO)}^\delta(\E)^{-1}$ for any $\M\in \OO$ due to the conditions in \eqref{eq:hypothesis_ct_general_input}, we have that
\begin{equation}\begin{aligned}
  \Theta(\M) &= R_{H,\aff(\OO)}^\delta(\E)^{-1} \tilde\N + (1 - R_{H,\aff(\OO)}^\delta(\E)^{-1} ) \Q\\
  &\propto \tilde\N + (R_{H,\aff(\OO)}^\delta(\E) - 1) \, \Q \in \OO'
\end{aligned}\end{equation}
and so $\Theta \in \SS$.
Noticing that $F(\Theta(\E), \N) \geq 1-\epsilon-2\delta$ analogously as in the previous case completes the proof.

\end{proof}

\begin{remark}
The quantity $R_{H,\OO}^\delta$ reduces to $R_{\min,\OO}$ when $\delta = 0$, as we mentioned before. However, for $R_{H,\aff(\OO)}^\delta$ and $R_{\min,\aff(\OO)}$, there is a potential gap between the two: we always have that $R_{H,\aff(\OO)}^\delta (\E) \geq 1$ for any $\delta$, but the quantity $R_{\min,\aff(\OO)}$ can be equal to 0 when there is no input state $\psi$ such that that $\<\Pi_{\idc \otimes \E(\psi)}, \idc \otimes \M(\psi) \>$ is constant for all $\M \in \OO$. Interestingly, either $R_{H,\aff(\OO)}^0$ or $R_{\min,\aff(\OO)}$ can be used in the proof in the case $\delta=0$ (in the sense that we can require either $R_{H,\aff(\OO)}^0(\E)\geq R_{\max,\OO'}^\epsilon(\N)$ or $R_{\min,\aff(\OO)} (\E) \geq R_{\max,\OO'}^\epsilon(\N)$), but using the hypothesis testing relative entropy provides a potentially more robust condition.
\end{remark}

{
\renewcommand\thetheorem{3}
\begin{theorem}
Let $\E\in\OO_{\rm all}$, $\N\in\OO_{\rm all}'$ and suppose $\idc\otimes\N(\psi)$ is pure for any pure state $\psi$.
If there exists a free superchannel $\Theta\in\SS$ such that $F(\Theta(\E),\N)\geq 1-\epsilon$, then it holds that $R_{H,\OO}^\epsilon(\E)\geq R_{\min,\OO'}(\N)$ and $R_{H,\aff(\OO)}^\epsilon(\E)\geq R_{\min,\aff(\OO')}(\N)$.

Conversely, there exists a free superchannel $\Theta\in\SS$ such that $F(\Theta(\E),\N)\geq 1 - \epsilon -\delta$ if $R_{H,\OO}^\epsilon(\E)\geq R_{s,\OO'}^\delta(\N)$ or if $R_{H,\aff(\OO)}^\epsilon(\E)\geq R_{\max,\OO'}^\delta(\N)$.
\end{theorem}
\addtocounter{theorem}{-1}
}
\begin{proof}

Let $\tilde\N$ be a channel such that $\Theta(\E) = \tilde\N$. 
Using the data processing inequality for hypothesis testing relative entropy~\cite{wang_2012,yuan_2019,gour_2019}, we get
\begin{equation}\begin{aligned}
  R_{H,\OO}^\epsilon (\E) &= \min_{\M\in \OO}  R^\epsilon_H(\E \| \M)\\
  &\geq \min_{\M\in \OO} R^\epsilon_H (\Theta(\E) \| \Theta(\M))\\
  &\geq \min_{\M' \in \OO'} R^\epsilon_H (\tilde\N \| \M')\\
  &= \min_{\M' \in \OO'}  \max_{\psi'} R^\epsilon_H (\idc \otimes \tilde\N (\psi') \| \idc \otimes \M' (\psi'))\\
  &= \min_{\M' \in \OO'}  \max_{\psi'} \max_{\substack{0 \leq P \leq \id\\\< P, \idc \otimes \tilde\N(\psi') \> \geq 1-\epsilon}} \<P, \idc \otimes\M' (\psi') \>^{-1}.
\end{aligned}\end{equation}
Using the assumption on $\N$, for all pure states $\psi'$ it holds that 
\bal
\< \idc \otimes \N(\psi'), \idc \otimes \tilde\N(\psi') \> \geq \min_{\psi'}\< \idc \otimes \N(\psi'), \idc \otimes \tilde\N(\psi') \> = F(\N, \tilde\N) \geq 1-\epsilon
\eal
and so we can always choose $P = \idc \otimes \N(\psi')$ as a feasible solution in the innermost maximization above. We thus get
\begin{equation}\begin{aligned}
  R_{H,\OO}^\epsilon (\E) &\geq \min_{\M' \in \OO'}  \max_{\psi'} \<\idc \otimes \N (\psi'), \idc \otimes\M' (\psi') \>^{-1}\\
  &= R_{\min,\OO'}(\N).
\end{aligned}\end{equation}

The proof for the condition using $R_{H,\aff(\OO)}$ follows in the same way since $R_{H}^\epsilon$ satisfies the data processing inequality even when the second argument is not a CP map (this is due to the fact that any feasible solution for $R^\epsilon _H (\E \| \M)^{-1}$ provides a feasible solution for $R^\epsilon_H (\Theta(\E) \| \Theta(\M))^{-1}$ regardless of what $\M$ is, and the objective function of $R^\epsilon_H (\E \| \M)^{-1}$ is linear in $\M$).

To show the second part, let $\tilde\N\in\OO_{\rm all}'$ be the channel such that $F(\tilde\N,\N)\geq1-\delta$ and $R_{s,\OO'}^\delta(\N)=R_{s,\OO'}(\tilde\N)$.
Then, the construction in \eqref{eq:full dimensional free superchannel} as well as the purity of $\idc \otimes \N(\psi)$ gives that 
\bal
F(\Theta(\E),\N)&= \min_{\psi} \<\idc\otimes\Theta(\E)(\psi),\idc\otimes\N(\psi)\>\\
&\geq\left<P^\star,\idc\otimes\E(\psi^\star)\right>\min_{\psi} \<\idc\otimes\tilde\N(\psi),\idc\otimes\N(\psi)\>\\
&\geq (1-\epsilon)(1-\delta)\\
&\geq 1-\epsilon-\delta.
\eal
The case of $R_{H,\aff(\OO)}$ can be shown analogously.

\end{proof}


\section{Remarks on distillable resource and resource cost}

In the main text, we defined the notions of one-shot distillable resource and resource cost as
\bal
  d^\epsilon_{\OO} (\E)\! &\coloneqq \sup \lset \mathfrak{R}_{\OO'}(\T) \!\bar\! F(\Theta(\E), \T) \geq 1\!-\!\epsilon,\; \T\! \in \TT,\; \Theta \in \SS \rset\\
  c^\epsilon_{\OO} (\E)\! &\coloneqq \inf \lset \mathfrak{R}_{\OO'}(\T) \!\bar\! F(\E, \Theta(\T)) \geq 1\!-\!\epsilon,\; \T\! \in \TT,\; \Theta \in \SS \rset
\eal
where $\TT$ is some given set of ``target channels'' and $\mathfrak{R}_{\OO'}$ is a choice of a monotone (that we fixed as $R_{\min,\OO'}$ for simplicity). In many practical settings, the set of targets is typically chosen simply as
\begin{equation}\begin{aligned}
  \TT = \lset \T^{\otimes m} \bar m \in \mathbb{N} \rset
\end{aligned}\end{equation}
for some fixed channel $\T$, such as a qubit resourceful channel. The quantities $d^\epsilon_\OO$ and $c^\epsilon_{\OO}$ then reduce to asking: what is the largest number of copies of $\T$ that can be distilled from $\E$, and what is the least number of copies of $\T$ required to synthesize $\E$ with free transformations? This corresponds exactly, for instance, to the familiar notion of distillable entanglement and entanglement cost known from entanglement theory.

Although we allow more general notions of target channels, it is clear that the set $\TT$ need not be a continuous set, which is precisely why we defined the notation
\begin{equation}\begin{aligned}
  \left\lfloor x \right\rfloor_{\TT} &\coloneqq \sup \lset \mathfrak{R}_{\OO'}(\T) \bar \mathfrak{R}_{\OO'}(\T) \leq x \rset\\
  \left\lceil x \right\rceil_{\TT} &\coloneqq \inf \lset \mathfrak{R}_{\OO'}(\T) \bar \mathfrak{R}_{\OO'}(\T) \geq x \rset.
\end{aligned}\end{equation}
This allows us to express the notions of distillable resource and resource cost in a concise form, and it reduces to the usual notions of floor $\lfloor\cdot\rfloor$ and ceiling $\lceil\cdot\rceil$ functions when $\mathfrak{R}_{\OO'}(\T)$ only takes values in $\mathbb{N}$.

\section{Details of the fidelity of distillation}

Here, we provide results and detailed discussions on the characterization of the maximal achievable fidelity, providing practical benchmarks on the distillation error. 
Recalling the definitions of $G_\OO$ and $G_{\aff(\OO)}$ introduced in \eqref{eq:G measure def}, we have the following result.

\begin{theorem} \label{thm:fidelity_distillation}
For a channel $\N\in\OO_{\rm all}'$, suppose that $\idc\otimes\N(\psi)$ is pure for any pure state $\psi$. Then, for any channel $\E\in\OO_{\rm all}$, it holds that  
\bal
 G_\OO(\E; R_{\min,\OO'}(\N)) \geq \max_{\Theta \in \SS} F(\Theta(\E), \N) \geq  G_\OO(\E; R_{s,\OO'}(\N))
 \eal
 and
 \bal
 G_{\aff(\OO)}(\E; R_{\min,\aff(\OO')}(\N)) &\geq \max_{\Theta \in \SS} F(\Theta(\E), \N)\geq  G_{\aff(\OO)}(\E; R_{\max,\OO'}(\N)).
 \eal
\end{theorem}

\begin{proof}
Let us consider $G_\OO$ first. Let $\Theta \in \SS$ be any free transformation, and $\psi^\star$ be a pure state such that $R_{\min, \OO'}(\N)^{-1} = \max_{\M \in \OO'} \< \idc \otimes \N (\psi^\star), \idc \otimes \M(\psi^\star) \>$.  Since $\Theta$ is a quantum superchannel, we can write it as $\Theta(\cdot) = \E_{\rm post} \circ \idc \otimes \cdot \circ \E_{\rm pre}$~\cite{chiribella_2008} for some CPTP maps $\E_{\rm pre}, \E_{\rm post}$. Define
\begin{equation}\begin{aligned}
  W \coloneqq [\idc \otimes \E^\dagger_{\rm post}] \left( \idc \otimes \N (\psi^\star) \right).
\end{aligned}\end{equation}
The complete positivity of $\E_{\rm post}$ immediately gives that $W \geq 0$, and coupled with the fact that $\E_{\rm post}$ preserves trace we get for any state $\sigma$ that
\begin{equation}\begin{aligned}
  \< W, \sigma \> &= \< \idc \otimes \N (\psi^\star), \idc \otimes \E_{\rm post} (\sigma) \> \leq 1
\end{aligned}\end{equation}
which means that $W \leq \id$. Finally, for any $\M \in \OO$, it holds that
\begin{equation}\begin{aligned}
  \< W, \idc \otimes \M \left(\idc \otimes \E_{\rm pre} (\psi^\star)\right) \> &= \< \idc \otimes \N (\psi^\star), \idc \otimes \Theta[\M] (\psi^\star) \>\\
  &\leq \max_{\M' \in \OO'} \< \idc \otimes \N (\psi^\star), \idc \otimes \M' (\psi^\star) \>\\
  &= R_{\min,\OO'}(\N)^{-1}.
\end{aligned}\end{equation}
Altogether, we have shown that $W$ and the state $\idc \otimes \E_{\rm pre} (\psi^\star)$ form a feasible solution to $G_\OO(\E; \ROmin(\N))$. Thus
\begin{equation}\begin{aligned}
  G_\OO(\E; \ROmin(\N)) &\geq \< W, \idc \otimes \E \left(\idc \otimes \E_{\rm pre} (\psi^\star)\right) \>\\
  &= \< \idc \otimes \N (\psi^\star), \idc \otimes \Theta[\E](\psi^\star) \>\\
  &\geq F(\Theta(\E), \N).
\end{aligned}\end{equation}
Since this holds for any $\Theta \in \SS$, the upper bound follows.

For the other inequality, take any feasible $W'$ and $\psi'$ in the optimization problem $G_\OO(\E; R_{s,\OO'}(\N))$, and define the superchannel
\begin{equation}\begin{aligned}
  \Theta(\L) \coloneqq \< W', \idc \otimes \L (\psi') \> \N + \< \id - W', \idc \otimes \L (\psi') \> \Q
\end{aligned}\end{equation}
where $\Q$ is a channel such that $\N  + (R_{s,\OO'}(\N)-1) \Q \in R_{s,\OO'}(\N) \OO'$. By an argument analogous to the one in the proof of Thm.~\ref{thm:transformation general}, we have that $\Theta \in \SS$ and get
\begin{equation}\begin{aligned}
  F(\Theta(\E), \N) &\geq \< W', \idc \otimes \E (\psi') \> \min_\psi\< \idc \otimes \N (\psi), \idc \otimes \N (\psi) \>\\
  &= \< W', \idc \otimes \E (\psi') \>.
\end{aligned}\end{equation}
Since this holds for any feasible $W'$ and $\psi'$, we get the desired inequality.

The proof for $G_{\aff(\OO)}$ proceeds analogously (cf.\ Thm.~\ref{thm:transformation general}).
\end{proof}

In the case of manipulating quantum states, Thm.~\ref{thm:fidelity_distillation} recovers the results of Ref.~\cite{regula_2020}.

\begin{remark}Our proofs can be straightforwardly adapted to a related notion of fidelity between channels often used in quantum communication theory, where the input state is fixed as the maximally entangled state $\psi^+$:
\begin{equation}\begin{aligned}
	\wt F(\E,\F) \coloneqq F( \idc \otimes \E (\psi^+), \idc \otimes \F (\psi^+) ).
\end{aligned}\end{equation}
If one is interested in approximate transformations with respect to this notion of channel distance, our main results can be recovered by replacing $R_{\min,\OO}$ with the quantity
\begin{equation}\begin{aligned}
	\wt R_{\min,\OO}(\E) &\coloneqq \min_{\M\in\OO} R_{\min}(\idc\otimes\E(\psi^+)||\idc\otimes\M(\psi^+)),
\end{aligned}\end{equation}
and analogously for $R_{\min,\aff(\OO)}$.
It is known that in settings such as quantum communication, the different notions of fidelity are asymptotically equivalent~\cite{kretschmann_2004}.
Interestingly, they can lead to the same quantitative statements already at the one-shot level: in the setting of no-signalling--assisted communication, we get from Thm.~\ref{thm:fidelity_distillation} that the maximal achievable fidelity is given by $G_{\OO_{\rm R}}(\E)$, and this exactly matches the fidelity computed in Ref.~\cite{Leung2015NS} using the alternative fidelity measure $\wt F$ (in fact, it is also equivalent to an analogous figure of merit in \textit{classical} communication with non-signalling assistance~\cite{Leung2015NS}).
\end{remark}


\section{Remarks on state-based resource theories}\label{app:state_based}

Some resource theories are concerned with the manipulation of channels in relation to some underlying state-based resource, characterized by a convex set of free states $\FF$, and the channels $\OO$ corresponds to some class of free operations for the state theory. In such cases, one can define monotones $R_{\max,\FF}, R_{\min,\FF}, R_{s,\FF}$ etc. in full analogy with the channel cases, but without the need to optimize over input states. There is then a natural choice of reference (target) channels: the preparation channel $\P_\phi$ with trivial input that simply prepares a chosen resourceful pure state $\phi$. Indeed, this is the way that entanglement~\cite{gour_2019,bauml_2019}, coherence~\cite{saxena_2020}, and more general channel resources based on states~\cite{liu_2020} are often approached.

Our results immediately apply to this setting. Because of the simpler structure of preparation channels, a number of simplifications will occur --- for instance, the quantification of the resourcefulness of $\P_\phi$ will reduce to monotones defined at the level of states.
\begin{lemma}\label{lem:replacement}
Let $\OO$ be any class of free operations such that $\P_\sigma \in \OO \; \forall \sigma \in \FF$. Then, for any preparation channel $\P_\phi$, it holds that
\begin{equation}\begin{aligned}
	\ROg(\P_\phi) &= \RFg(\phi)\\
	\ROs(\P_\phi)	&= \RFs(\phi)\\
	\ROmin(\P_\phi) &= \RFmin(\phi).
\end{aligned}\end{equation}
\end{lemma}
The Lemma follows immediately since the trivial input space of $\P_\phi$ forces all channels in the optimization to also be preparation channels, which means that we are optimizing over $\P_\sigma \in \OO$ s.t.\ $\sigma \in \FF$. 
The result can also be extended to more general replacement channels~\cite{regula_2020-2}.

As an immediate consequence, from Cors.~\ref{cor:dilution} and \ref{cor:distillation} we obtain that the one-shot distillable resource $d_\OO^\epsilon$ and one-shot resource cost of any channel can be quantified exactly, when the reference states $\phi$ satisfy $\RFs(\phi) = \RFmin(\phi)$ or $\RFg(\phi) = R_{\min,\aff(\FF)}(\phi)$. Notably, in Refs.~\cite{liu_2019,regula_2020}, it was shown that in every convex resource theory of quantum states there exists a set of so-called golden states $\phig$, such that each $\phig$ in the set maximizes both $\RFg$ and $\RFmin$ among all states in the given space, and crucially that $\RFg (\phig) = \RFmin (\phig)$. In many relevant resource theories, such golden states satisfy $\RFs(\phig) = \RFmin(\phig)$ --- e.g.\ in bi- and multipartite entanglement~\cite{vidal_1999,steiner_2003,regula_2020}, entanglement of Schmidt number $k$~\cite{johnston_2018}, multi-level coherence~\cite{johnston_2018}, and non-positive partial transpose --- or $\RFg(\phig) = R_{\min,\aff(\FF)}(\phig)$ --- e.g.\ in coherence and thermodynamics, thus fulfilling the required conditions.

Two direct applications of our results are then in the theory of bipartite channel entanglement~\cite{gour_2019,bauml_2019}, where the aim is to distill copies of the maximally entangled state, or to simulate the action of a given channel by using maximally entangled states. Specifically, the distillable entanglement and entanglement cost can be defined as
\bal
d^\epsilon_E(\E)&\coloneqq\max\lset\log d\sbar \exists \Theta\in\SS, F(\Theta(\E),\phi^+_d)\geq 1-\epsilon\rset\\
c^\epsilon_E(\E)&\coloneqq\min\lset\log d\sbar \exists \Theta\in\SS, F(\Theta(\phi^+_d),\E)\geq 1-\epsilon\rset,
\eal
where $\phi^+_d$ is the $d$-dimensional maximally entangled state, and the allowed free transformations $\SS$ will depend on which class of channels we wish to preserve. When the free channels are all separable channels $\OO_{\rm SEP}$, Cors.~\ref{cor:dilution} and \ref{cor:distillation} establish exact expressions for the distillable entanglement and entanglement cost of a channel under separability-preserving superchannels:
\begin{equation}\begin{aligned}
	d^\epsilon_E(\E) = \log \lfloor R_{H,\OO_{\rm SEP}}^\epsilon(\E) \rfloor, \qquad
	c^\epsilon_E(\E) = \log \lceil R_{s,\OO_{\rm SEP}}^\epsilon(\E) \rceil.
\end{aligned}\end{equation}
Analogously, when the free channels are all PPT channels $\OO_{\rm PPT}$, we get $d^\epsilon_E(\E)=\log\lfloor R_{H,\OO_{\rm PPT}}^\epsilon(\E) \rfloor$ and $c^\epsilon_E(\E)=\log \lceil R_{s,\OO_{\rm PPT}}(\E) \rceil$. These results generalize our considerations in quantum communication theory to the bipartite channel setting, and can be compared with the results of Refs.~\cite{gour_2019,bauml_2019} where the manipulation of channel entanglement under so-called PPT superchannels was considered. We note that the corresponding results in the special case of manipulating the entanglement of states were first reported in~\cite{brandao_2010,brandao_2011}.


\section{Remarks on the theory of nonlocality}\label{app:nonlocal}

For a number of settings for quantum nonlocality, the monotones discussed in this work can be straightforwardly computed by exploiting the fact that $\OO$ can be represented by linear or semidefinite constraints~\cite{Rosset2020type}. 
We use this to evaluate the measures for a class of boxes which includes the important cases of the Popescu-Rohrlich (PR) box and the Tsirelson box.

\begin{figure}[h!]
\centering
\includegraphics[width=7cm]{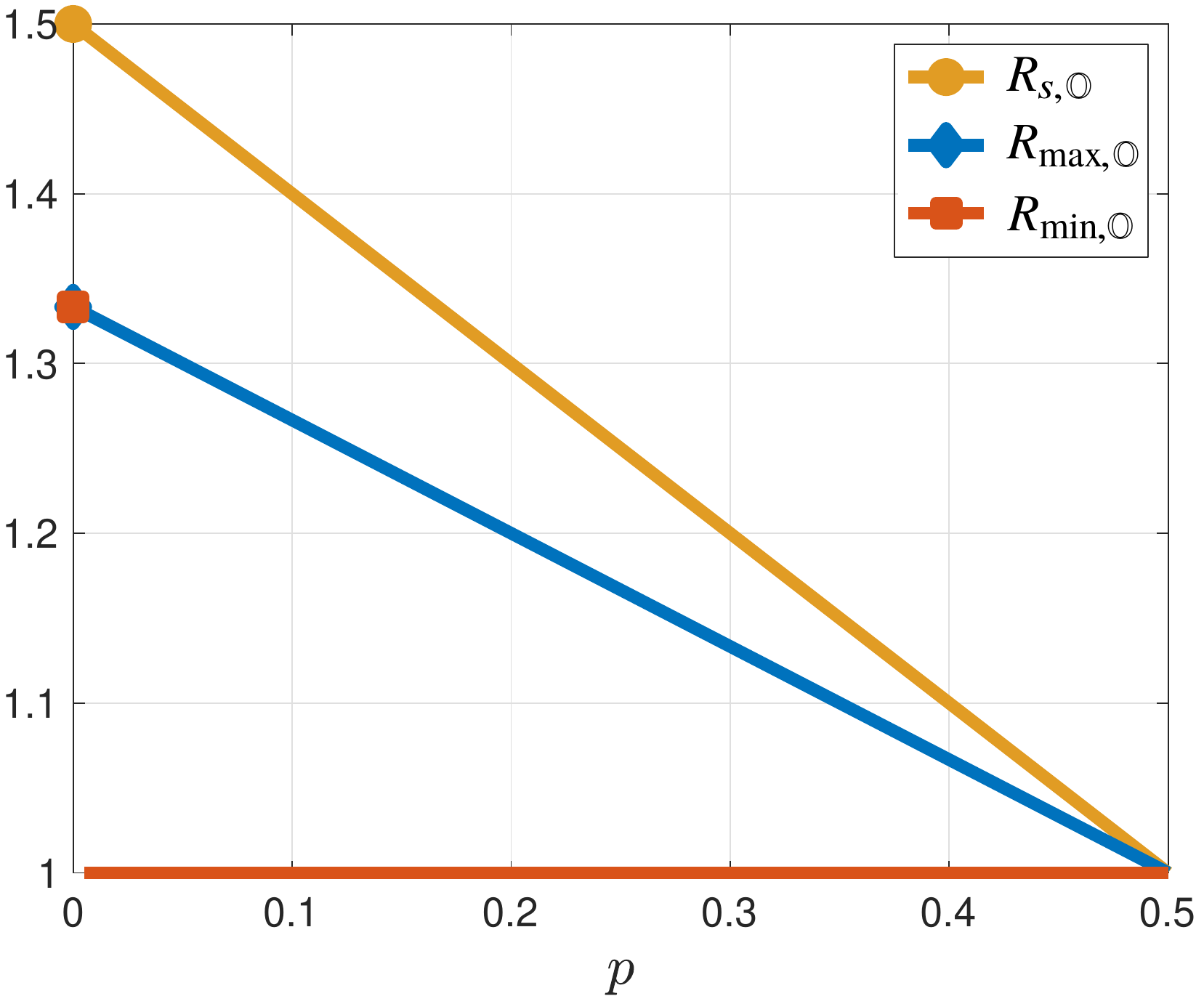}
\caption{Measures of nonlocality of the isotropic box $\B_p = (1-p) \B_{\rm PR} + p \B_{\rm random}$, where $\B_{\rm PR}$ is the PR box and $\B_{\rm random}$ denotes white noise. The Tsirelson box is obtained for $p = 1 - \frac{\sqrt{2}}{2} \approx 0.29$.}
\label{fig:nonlocal}
\end{figure}

We see in Fig.~\ref{fig:nonlocal} that the generalized robustness $\ROg$ linearly decreases from the value of $\ROg(\B_{\rm PR}) = \frac{4}{3}$, while $\ROs$ exhibits similar behavior from the value of $\ROs(\B_{\rm PR}) = \frac{3}{2}$. Notice in particular that $\ROg(\B_{\rm PR}) = \ROmin(\B_{\rm PR})$ (analogously to the property for golden states~\cite{liu_2019,regula_2020} in state-based theories), but the full-rank noise in the form of $\B_{\rm random}$ makes the min-entropy measure $\ROmin(\B_p)$ equal to 1 for all other boxes in this family. This means, in fact, that the transformation $\B_{p}^{\otimes n} \to \B_{\rm PR}$ with $p > 0$ is impossible for any $n$, since the Choi matrix $J_{\B_{p}^{\otimes n}}$ is always full rank and hence $\ROmin(\B_{p}^{\otimes n}) = 1$.


\section{Remarks on robustness measures}\label{app:robustness}

Due to their general formulation, robustness and other related measures have been considered in various settings such as nonlocality, contextuality, and measurement incompatibility, which impose specific structures on the resources of interest. 
We will show that the robustness measures defined based on the mixture of the original resources studied in previous works coincide with the channel-based robustness defined in our formalism.  
Here, we focus on the case of the generalized robustness of measurement incompatibility, but similar arguments can be applied to other types of resources as well as other variants of the measures such as the standard robustness and weight measure. 

Let $\mfM$ be the set of all sets of POVMs, and $\mfC$ be the set of compatible POVMs.
For any set of POVMs $\{M_{a|x}\}$, the robustness of measurement incompatibility was originally defined as~\cite{Haapasalo2015robustness,Uola2015robustness}
\bal
 R_\mfC(\{M_{a|x}\})\coloneqq\min\lset 1+r \sbar \frac{M_{a|x} + r N_{a|x}}{1+r}=F_{a|x}\,\forall a,x,\ \{F_{a|x}\}\in\mathfrak{C},\ \{N_{a|x}\}\in\mathfrak{M} \rset.
 \label{eq:def robustness measurement incompatibility original}
\eal

On the other hand, our formalism represents the sets of POVMs in the form of bipartite channels. 
Namely, the theory of measurement incompatibility can be formulated by taking $\OO_{\rm all}$ to be the set of channels representing the sets of POVMs as
\bal
 \OO_{\rm all}\coloneqq\lset\E_{\{M_{a|x}\}}\sbar\{M_{a|x}\}\in\mfM\rset
 \label{eq:POVMs to channel}
\eal
where we defined 
 \bal
 \E_{\{M_{a|x}\}}(\sigma\otimes \rho)\coloneqq \sum_{x,a} \bra{x}\sigma\ket{x} \<M_{a|x},\rho\>\dm{a}
 \eal
where $\{\ket{x}\}$ and $\{\ket{a}\}$ are orthonormal bases representing classical variables for measurement settings and measurement outcomes respectively. 
Then, the set of free channels $\OO\subseteq\OO_{\rm all}$ can naturally be defined using the set of compatible POVMs $\mfC$ in the place of $\mfM$ in \eqref{eq:POVMs to channel}. 

With this formalism, the robustness for channel $\E_{\{M_{a|x}\}}$ is given by \eqref{eq:robustness def}.
Below, we show the equivalence between this and the original robustness measure \eqref{eq:def robustness measurement incompatibility original}. 

\begin{lemma}
 Let $\OO_{\rm all}$ be the set of bipartite channels defined in \eqref{eq:POVMs to channel} and $\OO$ be the set of free channels that represent the compatible POVMs. Then, it holds that 
\bal
R_\mfC(\{M_{a|x}\})=R_{\OO}(\E_{\{M_{a|x}\}}).
\eal
\end{lemma}

\begin{proof}
Following the definition of $R_{\OO}$,
\bal
 R_{\OO}(\E_{\{M_{a|x}\}})&=\min\lset 1+r\sbar\frac{\E_{\{M_{a|x}\}} + r \N}{1+r}=\F,\ \N\in\OO_{\rm all},\F\in\OO \rset\\
 & =\min\lset 1+r\sbar\frac{\E_{\{M_{a|x}\}} + r \E_{\{N_{a|x}\}}}{1+r}=\E_{\{F_{a|x}\}},\ \{N_{a|x}\}\in\mfM, \{F_{a|x}\}\in\mfC \rset\\
 & =\min\lset 1+r \sbar \forall \sigma,\rho,\sum_{x,a} \bra{x}\sigma\ket{x} \<\frac{M_{a|x}+rN_{a|x}}{1+r}-F_{a|x},\rho\>\dm{a} =0,\ \{N_{a|x}\}\in\mfM, \{F_{a|x}\}\in\mfC \rset\\
 &=\min\lset 1+r \sbar \frac{M_{a|x} + r N_{a|x}}{1+r}=F_{a|x}\,\forall a,x,\ \{N_{a|x}\}\in\mfM,\{F_{a|x}\}\in\mathfrak{C} \rset
\eal
where the last equality is obtained by observing that in the expression of the the third line, the fact that the equality holds for any $\sigma$, as well as that $\dm{a}$ is a classical state, implies that $\<\frac{M_{a|x}+rN_{a|x}}{1+r}-F_{a|x},\rho\>=0,\forall a,x$, and further imposing that this holds for any $\rho$ gives $\frac{M_{a|x}+rN_{a|x}}{1+r}-F_{a|x}=0,\forall a,x$.
On the other hand, $\frac{M_{a|x}+rN_{a|x}}{1+r}-F_{a|x}=0,\forall a,x$ clearly implies the expression in the third line. 
The proof is completed by noting that the last expression is nothing but the one for $R_\mfC(\{M_{a|x}\})$.
\end{proof}

\end{document}